\documentclass[10pt,twocolumn,twoside]{IEEEtran} 
\usepackage{amssymb}

 \usepackage{enumitem}  
\usepackage{mathtools,amsthm}

\usepackage{graphics} 
\usepackage{epsfig}
\usepackage[linesnumbered,boxed,ruled,commentsnumbered]{algorithm2e}
\usepackage{amsmath,amssymb,amsfonts}
\usepackage{mathrsfs}
\usepackage{cite}
\usepackage{dsfont}
\usepackage{siunitx}
\usepackage{algorithmic}
\usepackage{graphicx}
\usepackage{textcomp}
\usepackage{cuted}
\usepackage{lipsum}
\usepackage{varioref}
\usepackage{import}
\usepackage{framed,xcolor}
\usepackage{color}
\usepackage{comment}
\usepackage{multirow}
\usepackage{epstopdf}   
\usepackage{psfrag}
\usepackage{breqn}
\usepackage{float}
\usepackage{mathtools}
\usepackage{booktabs}
\usepackage{soul}
\usepackage{amsbsy}
\usepackage[bottom]{footmisc}
\usepackage{lineno}
\usepackage{cleveref}
\usepackage{microtype}
\usepackage{comment}
\usepackage{epstopdf}

\newcommand{\vect}[1]{\ensuremath{\boldsymbol{\mathrm{#1}}}}
\newtheorem{theorem}{Theorem}

\newtheorem{Definition}{Definition}
\newtheorem{Remark}{Remark}

\newtheorem{Assumption}{Assumption}

\IEEEoverridecommandlockouts                              % This command is only needed if 
\title{\LARGE \bf
Intersection of Reinforcement Learning and Bayesian Optimization for Intelligent Control of Industrial Processes: A Safe MPC-based DPG using Multi-Objective BO}

\author{Hossein Nejatbakhsh Esfahani, Javad Mohammadpour Velni% <-this % stops a space
\thanks{*This work was supported by the US National Science Foundation under award \#2302219.}% <-this % stops a space
\thanks{H. N. Esfahani and J. M. Velni are with the Department of Mechanical Engineering, Clemson University, Clemson, SC, USA.
		{\tt\small \{hnejatb, javadm\}@clemson.edu%
		}.
}}

\begin{document}

\maketitle

\begin{abstract}
Model Predictive Control (MPC)-based Reinforcement Learning (RL) offers a structured and interpretable alternative to Deep Neural Network (DNN)-based RL methods, with lower computational complexity and greater transparency. However, standard MPC-RL approaches often suffer from slow convergence, suboptimal policy learning due to limited parameterization, and safety issues during online adaptation. To address these challenges, we propose a novel framework that integrates MPC-RL with Multi-Objective Bayesian Optimization (MOBO). The proposed MPC-RL-MOBO utilizes noisy evaluations of the RL stage cost and its gradient, estimated via a Compatible Deterministic Policy Gradient (CDPG) approach, and incorporates them into a MOBO algorithm using the Expected Hypervolume Improvement (EHVI) acquisition function. This fusion enables efficient and safe tuning of the MPC parameters to achieve improved closed-loop performance, even under model imperfections. A numerical example demonstrates the effectiveness of the proposed approach in achieving sample-efficient, stable, and high-performance learning for control systems.
\end{abstract}

\section{Introduction}
Reinforcement Learning (RL) is a powerful tool for tackling Markov Decision Processes (MDP) without relying on a model of the probability distributions underlying the state transitions of the real system \cite{sutton}. More precisely, most RL methods rely purely on observed state transitions, and realizations of the stage cost in order to enhance the performance of the control policy. RL seeks to determine an optimal policy through interaction with the environment. However, learning complex behaviors often demands a large number of samples, which can be impractical in real-world applications.

In contrast, actively selecting informative samples lies at the core of Bayesian Optimization (BO), which builds a probabilistic surrogate model of the objective function based on past evaluations to guide the selection of future samples. BO is a probabilistic optimization technique well-suited for optimizing black-box functions that are expensive to evaluate. It operates by constructing a surrogate model, typically a Gaussian Process (GP), to approximate the objective function and guide the search for optimal parameters \cite{7352306}. When optimization involves multiple conflicting objectives, Multi-Objective Bayesian Optimization (MOBO) extends this framework by modeling each objective with a separate surrogate and selecting evaluation points that improve the approximation of the Pareto front. Rather than seeking a single optimum, MOBO aims to identify a diverse set of trade-off solutions that are Pareto optimal, meaning no objective can be improved without degrading another \cite{mobo1,mobo2}.

The integration of Bayesian Optimization (BO) with Model predictive control (MPC) has recently emerged as a powerful framework for tackling complex control challenges in diverse application areas. MPC is often selected for its ability to handle both input and state constraints \cite{rawling}. In this context, \cite{HIRT2024208} introduced a safe and stability-aware BO framework to learn the cost functions of the MPC, where a parameterized MPC controller is adaptively tuned to maximize closed-loop performance despite model-plant mismatch. Furthermore, in \cite{KUDVA2024458}, a high-dimensional BO method was proposed for sample-efficient MPC tuning, enabling effective optimization in complex control settings. The authors in \cite{mobo3} leveraged MOBO to achieve performance-oriented model adaptation for an MPC scheme. To address multiple conflicting objectives in the adaptive control of water reservoir systems, a MOBO algorithm was used in \cite{mobo4} to automatically select the optimal weights for the MPC cost function. To tune an MPC scheme for the control of wind farms with a high-dimensional design space, a MOBO framework over sparse subspaces was proposed in \cite{mobo5}.

Due to uncertainties and unknown dynamics, accurate models of dynamical systems are often difficult to obtain. Even if accurate models are available, they may be in general too complex to be used in MPC schemes. Consequently, performance degradation often occurs due to inaccurate models used in the MPC schemes. Furthermore, choosing a suitable open-loop cost function and constraints to achieve the desired closed-loop performance while guaranteeing safety remains challenging. To mitigate the limitations of imperfect models, data-driven and machine learning (ML) techniques have been recently employed to enhance the accuracy of models used in MPC frameworks \cite{8909368,MullerAllgower2021,muntwiler22a,KARG2021107266}. However, a well-fitted model does not necessarily guarantee satisfactory closed-loop control performance, as control objectives may remain unmet despite accurate system representation. 

To overcome this challenge, recent works \cite{10808167,8701462,10542325,10644368} have explored the integration of Markov Decision Process (MDP) principles with MPC. By strategically adjusting the stage and terminal cost functions, these approaches allow the MPC framework to replicate the optimal policy of a corresponding Markov Decision Process (MDP). Building on this idea, recent studies have introduced MPC-based RL algorithms that aim to learn cost functions directly, with the goal of improving closed-loop performance. Rather than focusing solely on improving model accuracy, these methods leverage the parametric representation of the cost function as an alternative to model learning, motivated by the intrinsic link between model-based predictions and the structure of the MPC cost function.

The MPC-based RL method offers a more interpretable framework with lower complexity compared to deep neural network (DNN)-based RL algorithms. This is due to the flexible parameterization of MPC, which allows for the numerical design of its cost function, thereby facilitating an efficient and straightforward implementation. However, this approach still requires a large number of interactions with the environment to achieve convergence, which may be impractical for real-world control systems. Although MPC-based RL has the theoretical potential to recover an optimal policy given a sufficiently expressive parameterization, this condition is rarely met in practice. In real-world applications, parameterization is often too limited to capture the full complexity of the optimal solution. Moreover, simplistic choices for the value function or terminal cost may overlook important aspects of MPC tuning, further contributing to suboptimal performance. Practical challenges such as the local convergence behavior of RL algorithms can also hinder optimality. Overcoming these limitations typically requires well-chosen initializations, an observation that holds true across many RL-based approaches.

To address these issues, we propose a fusion of MPC-based RL and MOBO, utilizing noisy observations of the RL stage cost function along with its gradient estimates based on the Policy Gradient (PG) theorem \cite{silver2014deterministic}. This approach then enables a sample-efficient MPC-based RL method with fast convergence, while addressing safety concerns during the parameter updates of the MPC. To evaluate the gradient of the closed-loop performance a.k.a. the RL stage cost, we use a Compatible Deterministic Policy Gradient (CDPG) approach, in which a parameterized MPC scheme provides both the parametric policy and the parametric value function required by the CDPG. We then incorporate the evaluations of the RL stage cost and its gradient, obtained from the CDPG, into a MOBO algorithm that uses the Expected Hypervolume Improvement (EHVI) acquisition function to update the MPC parameters, aiming to safely achieve optimal closed-loop performance even in the presence of model imperfections, a.k.a model misspecification, in the underlying MPC scheme. Moreover, we leverage the proposed MOBO-based CDPG framework to enable a safe and stability-aware learning mechanism for MPC schemes.

This paper is organized as follows. Section \ref{sec:2} describes the MPC-based CDPG approach. Section \ref{sec:3} provides background on Bayesian Optimization and Multi-Objective Bayesian Optimization (MOBO). In Section \ref{sec:4}, we present the proposed MOBO-based RL method for learning a parameterized MPC scheme. A numerical example demonstrating the performance of the proposed method is given in Section \ref{sec:5}, and finally, the paper is concluded in Section \ref{sec:6}.

\section{MPC-based Reinforcement Learning}\label{sec:2}
Policy Gradient (PG) is a well-known policy-based reinforcement learning method that attempts to optimize a policy directly, rather than indirectly, through a value function. This section presents a Compatible Deterministic Policy Gradient (CDPG) based on a parameterized MPC, where the deterministic policy $\vect u_k=\vect\pi_{\vect\theta}\left(\vect x_k\right)$ delivered by the parameterized MPC scheme dedicates a deterministic action to each state $\vect x_k$ \cite{silver2014deterministic}. We then assume that the real system to be controlled is modeled as a discrete MDP with (potentially) stochastic state transition dynamics as
\begin{align}\label{eq:mdp_dyn}
    \vect x_{k+1}\sim\mathbb{P}[\cdot|\vect x_k,\vect u_k].
\end{align}
It is noted that the notation in \eqref{eq:mdp_dyn} is standard in the MDP literature, whereas the control literature typically adopts a different notation as
\begin{align}\label{eq:real_dyn}
    \vect x_{k+1}=\vect f(\vect x_k,\vect u_k,\vect d_k),
\end{align}
where $\vect{d}_k$ is a stochastic variable and $\vect{f}$ is a possibly nonlinear function. 
\subsection{MPC as an approximator for RL}
Let us denote the RL stage cost (baseline) associated with the MDP by $L\left(\vect x_k,\vect u_k\right)$, defined as
\begin{align}
    L\left(\vect x_k,\vect u_k\right)=
        l\left(\vect x_k,\vect u_k\right)+\mathcal{I}_\infty\left(\vect h\left(\vect x_k,\vect u_k\right)\right),
\end{align}
where $\vect h\left(\vect x_k,\vect u_k\right)\leq 0$ collects the inequality constraints. We use the following indicator function:
\begin{align}
   \mathcal{I}_\infty\left(x\right)=\left\{\begin{matrix}
\infty\quad \text{if}\quad x>0\\\quad0\qquad\text{otherwise}
\end{matrix}\right..
\end{align}
The value function associated with the true MDP then reads as
\begin{align}\label{eq:true_V}
     V^{\vect\pi}\left(\vect x_k\right)=\mathbb{E}\Big[\sum_{i=0}^\infty\gamma^{i} L\left(\vect x_i,\vect\pi\left(\vect x_i\right)\right)\Big],
\end{align}
where $\vect x_0=\vect x_k$, and $\gamma\in(0,1]$ denotes the discount factor. The optimal value function is obtained by $V^\star\left(\vect x_k\right)=\min_{\vect\pi} V^{\vect\pi}\left(\vect x_k\right)$. The optimal policy then reads $\vect\pi^\star\left(\vect s\right)=\arg\min_{\vect\pi} V^\star\left(\vect x_k\right)$. We next show that an MPC scheme can be a valid approximator in the context of RL such that the corresponding value function and policy can capture the true MDP.
\begin{Assumption}
    There exists a perfectly identifiable function $\vect\delta\left(\vect x_k,\vect u_k\right)$ that captures the stochasticity and uncertainty (model mismatch) of the real system \eqref{eq:real_dyn}, such that
    \begin{align}\label{eq:real_del}
        \vect x_{k+1}=\vect f(\vect x_k,\vect u_k,\vect d_k)=\vect{\hat f}(\vect x_k,\vect u_k)+ \vect\delta\left(\vect x_k,\vect u_k\right),
    \end{align}
    where $\vect{\hat f}(\vect x_k,\vect u_k)$ denotes an imperfect model of \eqref{eq:real_dyn}.
\end{Assumption} 
\begin{Remark}
   A wide range of existing approaches aim to accurately identify the model mismatch $\vect\delta\left(\vect x_k,\vect u_k\right)$ in \eqref{eq:real_del} by leveraging Machine Learning (ML) algorithms and uncertainty quantification techniques. Despite their predictive capabilities, ML-based MPC models often lack an explicit connection to the underlying control objectives. In most cases, these models are trained solely to minimize prediction error, with the implicit assumption that better predictions will naturally translate into improved MPC performance, an assumption that may not hold in practice.
\end{Remark}
Under Assumption 1, we define the value function associated with a finite version of the true MDP \eqref{eq:true_V} as
\begin{align}\label{eq:V_finite_MDP}
    &V^{\vect\pi,N}\left(\vect x_k\right)=\mathbb{E}\Bigg[L\left(\vect x_0,\vect\pi\left(\vect x_0\right)\right)\\\nonumber
    &+\gamma^N T\left(\vect{\hat f}(\vect x_{N-1},\vect u_{N-1})+\vect\delta\left(\vect x_{N-1},\vect u_{N-1}\right)\right)+\\\nonumber
    &+\sum_{i=1}^{N-1}\gamma^i L\left(\vect{\hat f}(\vect x_{i-1},\vect u_{i-1})+ \vect\delta\left(\vect x_{i-1},\vect u_{i-1}\right),\vect u_i\right)\\\nonumber
    &\qquad\qquad\qquad\qquad\qquad\qquad\Big|\vect x_0=\vect x_k,\quad \vect u_i=\vect\pi\left(\vect x_i\right)\Bigg],
\end{align}
where $T$ and $L$ denote the terminal and stage cost functions, respectively.
\begin{Remark}
    In this paper, we demonstrate that the terminal and stage cost functions, $T$ and $L$ in \eqref{eq:V_finite_MDP}, can be modified without the need to identify a perfectly accurate function $\vect\delta$ in \eqref{eq:real_del}, such that the resulting value function approximates $V^\star\left(\vect x_k\right)$.
\end{Remark}
According to Remark 1, we introduce the modified terminal cost function $T_{\vect\delta}$ and the modified stage cost function $L_{\vect\delta}$, such that the corresponding finite-time value function is given by:
\begin{align}\label{eq:modif_V_finite}
    &\hat V^{\vect\pi,N}\left(\vect x_k\right)=\mathbb{E}\Bigg[\gamma^N T_{\vect\delta}\left(\vect {\hat x}_N\right)+L\left(\hat{\vect {x}}_0,\vect\pi\left(\hat{\vect {x}}_0\right)\right)\\\nonumber
    &\qquad+\sum_{i=1}^{N-1}\gamma^i L_{\vect\delta}\left(\vect {\hat x}_{i},\vect {\hat u}_{i}\right)\Big|\hat{\vect {x}}_0=\vect x_k,\quad \hat{\vect {u}}_i=\vect\pi\left(\hat{\vect {x}}_i\right)\Bigg],
\end{align}
where $\vect {\hat x}_{i+1}=\vect{\hat f}(\vect {\hat x}_{i},\vect {\hat u}_{i})$. Inspired by \cite{10808167}, we define the modified stage cost function as
\begin{align}\label{eq:modif_L}
     L_{\vect{\delta}}(\hat{\vect{x}}_i, \vect{\pi}(\hat{\vect{x}}_i)) = V^\star(\hat{\vect{x}}_i) - \gamma \mathbb{E}[V^\star(\hat{\vect{x}}_{i+1})].
\end{align}
More precisely, this choice is made to account for model mismatch and provides an effective way to modify the control objective for the following reasons:

    \begin{itemize}
        \item \textbf{Reflection of Model Mismatch:} 
            The term \(\vect\delta(\hat{\vect{x}}_i, \vect{\pi}(\hat{\vect{x}}_i))\) in the real system dynamics equation (\ref{eq:real_del}) captures the difference between the true system dynamics and the model's prediction. By using \eqref{eq:modif_L}, we are directly incorporating the impact of model mismatch on the control policy. The difference between the true system value function and the model-predicted value function is a measure of the discrepancy introduced by the mismatch, and this term helps to adjust for that.
        
        \item \textbf{Incorporating Long-Term Consequences:} 
            The term \(V^\star(\hat{\vect{x}}_i)\) represents the long-term cost starting from state \(\hat{\vect{x}}_i\), while \(\gamma \mathbb{E}[V^\star(\hat{\vect{x}}_{i+1})]\) represents the expected long-term cost starting from the predicted next state \(\hat{\vect{x}}_{i+1}\). By minimizing this difference, the control policy adapts to ensure that the impact of model mismatch is minimized over time, ensuring more accurate decision-making over the prediction horizon.
    \end{itemize}

\begin{Assumption}
    To ensure that the modified stage cost function \eqref{eq:modif_L} reflects the model mismatch effectively, we make the following assumptions:
\begin{enumerate}
\item \textbf{Model Mismatch is Bounded:} 
            The model mismatch \(\vect\delta(\hat{\vect{x}}_i, \vect{\pi}(\hat{\vect{x}}_i))\) is bounded, i.e., there exists a small upper bound \(\epsilon_{\delta}\) such that $
            \|\vect\delta(\hat{\vect{x}}_i, \vect{\pi}(\hat{\vect{x}}_i))\| \leq \epsilon_{\delta}$. This ensures that the model mismatch does not grow unboundedly, and the correction made by the stage cost function remains effective.

\item \textbf{Smooth and Well-Defined Value Function:}
            The value function \(V^\star(\hat{\vect{x}}_i)\) is assumed to be smooth and well-defined. This ensures that the correction made by \eqref{eq:modif_L} is meaningful. If the value function is discontinuous or undefined, this modification would not provide a reliable correction for the model mismatch. We then assume that there exists a non-empty set $\mathcal{X}_0 \subseteq \mathcal{X}$ such that, for all $\hat{\vect{x}} \in \mathcal{X}0$ and all $\gamma \in (0,1]$, the following holds:
\begin{align}
&\left|\gamma^{k}\mathbb{E}\left[V^\star\left(\hat{\vect{x}}_{i}\right)\right]\right| < \infty, \quad \forall\ 0 < i < N.
\end{align}

 \item \textbf{Discount Factor Close to One:}
            The discount factor \(\gamma\) should be close to 1. This is because the modification relies on considering the long-term consequences of the model mismatch. A small \(\gamma\) would discount future values too heavily, making it less effective at correcting for model mismatch over the long term.

\end{enumerate}
\end{Assumption}

\begin{Remark}
    The modified stage cost function \eqref{eq:modif_L} is a natural choice for reflecting model mismatch when the above assumptions hold. This formulation ensures that the control policy can adapt to the model mismatch in a way that minimizes its impact on long-term performance, thereby improving the overall robustness and accuracy of the system's decision-making process.
\end{Remark}

\begin{theorem}\label{theorem:V_modif}
    An Optimal Control Problem (OCP) with the associated value function \eqref{eq:modif_V_finite} can yield the optimal value function $V^\star$ associated with the true MDP, such that the following holds:
    \begin{align}
        \hat V^\star\left(\vect x_k\right)=\min_{\vect\pi}\hat V^{\vect\pi,N}\left(\vect x_k\right)=V^\star\left(\vect x_k\right),
    \end{align}
    and the corresponding optimal policy can capture the true optimal policy $\vect\pi^\star$.
\end{theorem}
\begin{proof}
Let us consider the modified stage cost \eqref{eq:modif_L} and choose the terminal cost $T_{\vect{\delta}}(\hat{\vect{x}}_N) = V^\star(\hat{\vect{x}}_N)$. Now, we apply a telescoping sum to the modified value function \(\hat{V}^{\vect{\pi}, N}(\vect{x}_k)\) in \eqref{eq:modif_V_finite}. Using the modified cost functions, we can express the total cost over the horizon as
\begin{align}\label{eq:modifV}
    &\hat V^{\vect\pi,N}\left(\vect x_k\right)=\mathbb{E}\Bigg[\gamma^N V^\star(\hat{\vect{x}}_N)+L\left(\hat{\vect {x}}_0,\vect\pi\left(\hat{\vect {x}}_0\right)\right)\\\nonumber
    &\qquad+\sum_{i=1}^{N-1}\gamma^i\left(V^\star(\hat{\vect{x}}_i) - \gamma \mathbb{E}[V^\star(\hat{\vect{x}}_{i+1})]\right)\Big|\hat{\vect {x}}_0=\vect x_k\Bigg].
\end{align}
Let us define the modified term $V_m$ as
\begin{align}
    V_m=\sum_{i=1}^{N-1}\gamma^i\left(V^\star(\hat{\vect{x}}_i) - \gamma \mathbb{E}[V^\star(\hat{\vect{x}}_{i+1})]\right).
\end{align}
Under Assumption 2, this modified term then becomes a telescoping sum such that
\begin{align}\label{eq:modif_term}
    V_m=\gamma V^\star(\hat{\vect{x}}_1)-\gamma^N\mathbb{E}[V^\star(\hat{\vect{x}}_N)].
\end{align}
Substituting \eqref{eq:modif_term} in \eqref{eq:modifV}, we have that
\begin{align}
    &\hat V^{\vect\pi,N}\left(\vect x_0\right)=L\left(\hat{\vect {x}}_0,\vect\pi\left(\hat{\vect {x}}_0\right)\right)+\gamma\mathbb{E}[V^\star(\hat{\vect{x}}_{1})]\\\nonumber
&\qquad\qquad\qquad=Q^\star(\hat{\vect{x}}_{0},,\vect\pi\left(\hat{\vect {x}}_0\right))=Q^\star({\vect{x}}_{k},\vect\pi\left({\vect {x}}_k\right)).
\end{align}
It implies that
\begin{align}
    \hat V^\star\left(\vect x_k\right)&=\min_{\vect\pi}\hat V^{\vect\pi,N}\left(\vect x_k\right)\\\nonumber
    &=\min_{\vect\pi}Q^\star({\vect{x}}_{k},\vect\pi\left({\vect {x}}_k\right))=V^\star\left(\vect x_k\right),
\end{align}
and ${\vect{\pi}}^\star\left(\vect x_k\right)=\arg\min_{{\vect{\pi}}} \hat V^\star\left(\vect x_k\right)$.
\end{proof}
It is important to note that the central Theorem \ref{theorem:V_modif} establishes the existence of a cost modification and offers insight into its structure. However, this modification is not directly implementable in practice, as the modified cost function is not known a priori. To overcome this challenge, we propose leveraging Reinforcement Learning (RL) to adjust a parametric version of the MPC cost functions, parameterized by $\vect\theta$, as detailed in the remainder of the paper.
\subsection{Compatible DPG}
We define the closed-loop performance of a parameterized policy $\vect\pi_{\vect\theta}$ for a given stage cost $L\left(\vect x_k,\vect u_k\right)$ as the following total expected cost, a.k.a policy performance index.
\begin{align}\label{eq:J}
J\left(\vect\pi_{\vect\theta}\right)=\mathbb{E}\Bigg[\sum_{k=0}^{\infty}\gamma^kL\left(\vect x_k,\vect u_k\right)\Bigg|\vect u_k=\vect \pi_{\vect\theta}(\vect x_k)\Bigg],
\end{align}
where the expectation $\mathbb{E}$ is taken over the distribution of the Markov chain in the closed-loop system under policy $\vect\pi_{\vect\theta}$. The policy parameters $\vect\theta$ can be directly optimized using gradient descent to minimize the expected closed-loop cost achieved by executing the policy $\vect\pi_{\vect\theta}$. 
\begin{align}
\label{eq:theta}
    \vect\theta \leftarrow \vect\theta-\alpha  \nabla _{\vect\theta}J(\vect\pi _{\vect\theta}),
\end{align}
where $\alpha>0$ is the learning rate. The policy gradient then reads as
\begin{align}\label{eq:PG_Q}
    \nabla_{\vect\theta} J\left(\vect\pi_{\vect\theta}\right)=\mathbb{E}\left[\nabla_{\vect\theta}\vect\pi_{\vect\theta}\left(\vect x_k\right)\nabla_{\vect u}Q^{\vect\pi_{\vect\theta}}\left(\vect x_k,\vect u_k\right)\Big|_{\vect u_k=\vect\pi_{\vect\theta}\left(\vect x_k\right)}\right],
\end{align}
where the expectation is taken over the trajectories of the real system subject to policy ${\vect{ \pi}_{\vect\theta}}$. Note that $\nabla_{\vect u}Q^{\vect\pi_{\vect\theta}}\left(\vect x_k,\vect u_k\right)$ can be replaced by $\nabla_{\vect u}A^{\vect\pi_{\vect\theta}}\left(\vect x_k,\vect u_k\right)$, where $A^{\vect\pi_{\vect\theta}}\left(\vect x_k,\vect u_k\right)=Q^{\vect\pi_{\vect\theta}}\left(\vect x_k,\vect u_k\right)-V^{\vect\pi_{\vect\theta}}\left(\vect x_k\right)$ denotes the advantage function. A necessary condition of optimality to $\vect\pi_{\vect\theta}$ then reads as
\begin{align}\label{eq:necessary_condition}
    \nabla_{\vect\theta} J\left(\vect\pi_{\vect\theta}\right)=0.
\end{align}
In the context of CDPG, one can use a \textit{compatible} approximation of the action-value function $Q^{\vect\pi _{\vect\theta}} (\vect x_k,\vect u_k)$ for which there is a class of compatible function approximators $Q^{\vect w}(\vect x_k,\vect u_k)$ such that the policy gradient is preserved. The compatible state-action function then reads as
\begin{align}
\label{eq:Q_w}
&Q^{\vect w}(\vect x_k,\vect u_k)=\underbrace{{{\left( {\vect u_k - {\vect\pi _{\vect\theta}\left(\vect x_k\right) }} \right)}^{\top}}\nabla_{\vect\theta}\vect\pi_{\vect\theta}^\top\left(\vect x_k\right){\vect w}}_{A^{\vect w}} + { V^{\vect\nu}}\left( \vect x_k\right),
\end{align}
where $A^{\vect w}$ denotes the advantage function parameterized by $\vect w$, approximating $A^{\vect w}\approx A^{\vect\pi_{\vect\theta}}$. The second term is a baseline function parameterized by $\vect\nu$, providing an approximation of the value function $V^{\vect\nu}\approx V^{\vect\pi_{\vect\theta}}$. Both functions can be computed by the linear function approximators as
\begin{align}
       &{V^{\vect\nu}} \left(\vect x_k\right ) =\vect\Upsilon\left(\vect x_k \right)^\top {\vect \nu},\quad {A^{\vect w}} \left(\vect x_k,\vect u_k\right ) =\vect \Psi\left(\vect x_k,\vect u_k \right)^\top {\vect w}\label{eq:adv},
\end{align}
where $\vect\Upsilon\left(\vect x_k\right)$ is the feature vector, and $\vect \Psi_k:={{\left( {\vect u_k - {\vect\pi_{\vect\theta}\left(\vect x_k\right) }} \right)}^{\top}}\nabla_{\vect\theta}\vect\pi_{\vect\theta}^\top\left(\vect x_k\right)$ includes the state-action features. The parameters $\vect w$ and $\vect\nu$ of the action-value function approximation then become the solutions of the following Least Squares (LS) problem
\begin{align}
\label{eq:error}
    \min_{\vect w, \vect\nu} \mathbb{E} \left[\big( Q^{\vect\pi_{\vect\theta}}(\vect x_k,\vect u_k)-Q^{\vect w} (\vect x_k,\vect u_k)\big )^2\right].
\end{align}
The problem above can be solved via the Least Square Temporal Difference (LSTD) method, which belongs to \textit{batch methods}, seeking to find the best fitting value function and action-value function, and it is more sample efficient than other methods.
\subsection{MPC-based CDPG}
Let us formulate a parameterized MPC scheme as follows:
	\begin{subequations}\label{eq:V0}
		\begin{align}
			V_{{\vect\theta} }(\vect x_k)=&\min_{\hat {\vect x},\hat {\vect u},\vect \eta }\quad \gamma^{N}\left(V^f_{\vect\theta}(\hat {\vect x}_N)+\vect \Gamma_{f}^\top\vect\eta_{N}\right) \nonumber\\
			&\quad +\sum_{i=0}^{N-1}\gamma^{i}\left(l_{\vect\theta}(\hat {\vect x}_i,\hat {\vect u}_i)+\vect \Gamma^\top\vect \eta_{i}\right)\label{eq:cost_mpc}\\
			\mathrm{s.t.}
                &\quad \hat {\vect x}_{i+1}=\vect {\hat f}_{\vect\theta}(\hat {\vect x}_i,\hat {\vect u}_i),\label{eq:mpc_model}\\
			&\quad\hat {\vect x}_0=\vect x_k, \label{eq:v2}\\
			&\quad \vect g(\hat {\vect u}_{i})\leq 0,\\
			&\quad \vect h_{\vect\theta}(\hat {\vect x}_{i},\hat {\vect u}_i)\leq \vect \eta_{i},\quad \vect h_{\vect\theta}^{f}(\hat {\vect x}_{N})\leq \vect \eta_{N}, \label{eq:violation} \\
			&\quad \vect\eta_{0,\ldots,N} \geq 0, \label{eq:slack}
		\end{align}
	\end{subequations}
where the parametric functions $l_{\vect\theta},V^f_{\vect\theta},\vect {\hat f}_{\vect\theta},\vect h_{\vect\theta},\vect h_{\vect\theta}^{f}$ are the stage cost, the terminal cost, the MPC model, the mixed constraints and the terminal constraints, respectively. $\vect g$ denotes the pure input constraints. In many real processes, there are uncertainties and disturbances that may cause an MPC scheme to become infeasible. Therefore, an $\ell_1$ relaxation of the mixed constraints \eqref{eq:violation} is introduced. An exact penalty is then imposed on the corresponding slack variables $\vect \eta_i$ with large enough weights $\vect \Gamma,\vect \Gamma_f$ such that the trajectories predicted by the MPC scheme will respect the constraints. All elements in the above MPC scheme are parameterized by $\vect\theta$, which will be adjusted by RL. The policy at the physical current time $k$ then reads as
	\begin{gather}\label{eq:Policy}
		\vect\pi_{\vect\theta}(\vect x_k)=\hat {\vect u}_0^\star\left(\vect x_k,\vect\theta\right),
	\end{gather}
where $\hat {\vect u}_0^{\star}$ is the first element of the input sequence $\hat {\vect u}_0^{\star},\cdots,\hat {\vect u}_{N-1}^{\star}$ solution of \eqref{eq:V0}. We next consider this optimal policy delivered by the MPC scheme as an action $\vect a$ in the context of reinforcement learning, where it is selected according to the above policy with the possible addition of exploratory moves. In an MPC-based DPG algorithm, two terms are needed, including a compatible action value function $Q^{\vect w}(\vect x_k,\vect u_k)$ and a policy gradient $\nabla_{\vect\theta} J\left(\vect\pi_{\vect\theta}\right)$. These terms are constructed based on the policy sensitivity term $\nabla_{\vect\theta}\vect\pi_{\vect\theta}\left(\vect x_k\right)$, which can be calculated using a sensitivity analysis on the MPC scheme. To this end, let us define the primal-dual Karush Kuhn Tucker (KKT) conditions underlying the MPC scheme \eqref{eq:V0} as 
\begin{align}
    \vect R = {\left[ {\begin{array}{*{20}{c}}
{{\nabla _{\vect \zeta}}{\mathcal L_{\vect\theta} }},{{\vect G_{\vect\theta} }},{\mathrm{diag}\left(\vect\mu\right) \vect H_{\vect\theta} }
\end{array}} \right]^\top},
\end{align}
where $\vect\zeta_k=\{\vect x,\vect u,\vect\eta\}$ includes the primal decision variables of \eqref{eq:V0} and the term $\mathcal{L}_{\vect \theta}$ is the associated Lagrange function as follows:
\begin{align}
\mathcal{L}_{\vect \theta}(\vect y_k) = \Phi_{\vect \theta} + \vect\lambda^\top \vect G_{\vect\theta}  + \vect\mu^\top \vect H_{\vect \theta},
\end{align}
where $\Phi_{\vect\theta}$ is the total MPC cost \eqref{eq:cost_mpc}, $\vect G_{\vect\theta}$ gathers the equality constraints and $\vect H_{\vect\theta}$ collects the inequality constraints of the MPC \eqref{eq:V0}. Let $\vect\lambda,\vect\mu$ be the associated dual variables. Argument ${\vect y_k}$ reads as ${\vect y_k} =\{\vect\zeta,\vect\lambda,\vect\mu\}$ and $ {\vect y}_k^{\star}$ refers to the solution of the MPC \eqref{eq:V0}. Consequently, the policy sensitivity ${\nabla _{\vect \theta} }{\vect \pi _{\vect \theta} }$ can then be obtained as \cite{10644368}
\begin{align}\label{eq:nabla_pi}
{\nabla _{\vect \theta} }{\vect \pi _{\vect \theta} }\left(\vect  x_k \right) =  - {\nabla _{\vect\theta} }{\vect R }\left( {\vect y_k^\star},\vect x_k,\vect\theta\right){\nabla _{\vect y_k}}{\vect R }{\left( {\vect y^\star_k},\vect x_k,\vect\theta \right)^{ - 1}}\frac{\partial {\vect y_k}}{\partial {\vect u_k}}.
\end{align}
\section{Background on Bayesian Optimization}\label{sec:3}
Bayesian Optimization (BO) is an efficient approach for optimizing complex, expensive-to-obtain, and noisy black-box functions,
\begin{align}
   \vect{\theta}^\star = \arg \min_{\vect{\theta}\in\mathcal{P}} J(\vect{\theta}), 
\end{align}
where $\vect\theta^\star$ is the global minimizer of $J(\vect{\theta})$ on the parameter space $\mathcal{P}$. In control applications, it is particularly useful for learning parameters in controller such as MPC.
\subsection{Gaussian Process}
The BO methods use Gaussian Process (GP) surrogate models to learn an approximation of not-explicitly-known functions with respect to some parameters, e.g., the closed-loop performance of a parameterized MPC $J\left(\vect\theta\right):=J\left(\vect\pi_{\vect\theta}\right)$ with respect to the policy parameters $\vect\theta$. Considering \( J: \mathbb{R}^{n_{\vect\theta}} \rightarrow \mathbb{R} \), we then define the corresponding GP model as
\begin{align}\label{eq:GP}
  J(\vect\theta) \sim \mathcal{GP}(\mu(\vect\theta), k(\vect\theta, \vect\theta')),
\end{align}
where $\mu(\vect{\theta}) = \mathbb{E}[J(\vect{\theta})]$ is the mean function, typically assumed to be zero, and the kernel $k(\vect \theta, \vect \theta')$ is the covariance function, determining the relationship between points in the input space. More precisely, the covariance function describes how the values of the function relate to each other. A commonly used kernel is the squared exponential function as
\begin{align}
    k(\vect{\theta}, \vect{\theta'}) = \sigma_J^2 \exp\left(-\frac{1}{2l^2} ||\vect{\theta} - \vect{\theta'}||^2\right),
\end{align}
where $\sigma_J^2$ is the signal variance, and $l$ is the length scale. The choice of kernel significantly impacts the behavior of the GP model. Given observed data \(\mathcal{D}^k = \{(\vect{\theta}^i, y^i)\}_{i=1}^k,\quad y^i = J(\vect{\theta}^i) + \epsilon^i\) and \(\epsilon^i \sim \mathcal{N}(0, \sigma_n^2)\), we then have
\[
\mathbf{y} \sim \mathcal{N}(\mu(\vect{\Theta}), K + \sigma_n^2 I),
\]
where $\vect\Theta=\left\{\vect\theta^i\right\}_{i=1}^k$ and $K$ denotes the covariance matrix ($n\times n$ kernel matrix with elements $[K]_{(i,j)}=k(\vect\theta_i,\vect\theta_j)$) computed over the training inputs. The posterior distribution of \(J(\vect{\theta}^\star)\) at a new input \(\vect{\theta}^\star\) is then given by:
\begin{align}\label{eq:post}
\begin{bmatrix}
\mathbf{y} \\
J^\star 
\end{bmatrix}
\sim \mathcal{N}\left(\begin{bmatrix}
\mu(\vect{\Theta}) \\
\mu(\vect{\theta}^\star)
\end{bmatrix}, \begin{bmatrix}
K + \sigma_n^2 I & k(\vect{\theta}^\star, \vect{\Theta}) \\
k(\vect{\Theta}, \vect{\theta}^\star) & k(\vect{\theta}^\star, \vect{\theta}^\star)
\end{bmatrix}\right).
\end{align}
This allows us to derive the predictive mean and variance for $J(\vect{\theta}^\star)$:
\begin{align}\label{eq:mean_var}
&\mu(\vect{\theta}^\star) = k(\vect{\theta}^\star, \vect{\Theta})(K + \sigma_n^2 I)^{-1} \mathbf{y},\\\nonumber
&\sigma^2(\vect{\theta}^\star) = k(\vect{\theta}^\star, \vect{\theta}^\star) - k(\vect{\theta}^\star, \vect{\Theta})(K + \sigma_n^2 I)^{-1}k(\vect{\Theta}, \vect{\theta}^\star).
\end{align}
The updated GP models are then used to induce an acquisition function $\alpha(\vect{\theta})$, which helps find the optimum of the unknown objective function $J(\vect\theta)$. 
\begin{align}
   \vect{\theta}^\star = \arg \min_{\vect{\theta}\in\mathcal{P}} \alpha(\vect{\theta}).
\end{align}
The core idea behind the acquisition function $\alpha(\vect{\theta})$ is to leverage the uncertainty information provided by the probabilistic surrogate of the objective function(s) to systematically balance exploration and exploitation of the design space, enabling the identification of an \textit{optimal} combination of decision variables within a limited number of evaluations. In this paper, we use the Expected Improvement (EI) as an acquisition function for a single-objective BO.
\begin{align}\label{eq:EI}
    \alpha_{\text{EI}}(\vect{\theta}) = \mathbb{E}[\min(J(\vect{\theta})-J^\star, 0)],
\end{align}
where $J^\star$ is the best observed function value. This acquisition function can be computed using the properties of the Gaussian distribution as
\begin{align}\label{eq:acq_EI}
  \alpha_{\text{EI}}(\vect{\theta}) =-\left( \left(J^\star - \mu(\vect{\theta})\right) \Phi(Z) + \sigma(\vect{\theta}) \phi(Z)\right),  
\end{align}
where $Z = \frac{J^\star - \mu(\vect{\theta})}{\sigma(\vect{\theta})}$. $\Phi(Z)$ and $\phi(Z)$ denote the standard normal Cumulative Density Function (CDF) and Probability Density Function (PDF) of the standard normal distribution, respectively.
\subsection{Multi-Objective BO}
In multi-objective optimization, we aim to minimize multiple objectives:
\begin{align}\label{eq:objs}
   \vect{f}(\vect{\theta}) = [f_1(\vect{\theta}), f_2(\vect{\theta}), \ldots, f_M(\vect{\theta})], 
\end{align}
where $\vect{f}(\vect{\theta}) \in \mathbb{R}^M,M\geq 2$ represents the performance space. Objectives are often conflicting, resulting in a set of optimal solutions, rather than a single best solution. These optimal solutions are referred to as \textit{Pareto set} $P_s\subseteq\mathcal{P}$ in the parameter space, and the corresponding images in performance space are \textit{Pareto front} $\tilde P=\vect f(P_s)\subset\mathbb{R}^M$. 
\begin{Definition}
A solution $\vect\theta^\star\in P_s$ is considered \textit{Pareto-optimal} if there is no other point $\vect\theta\in\mathcal{P}$ such that 
\begin{align}
   &f_i(\vect{\theta}^\star) \leq f_i(\vect{\theta}) \quad \forall i \in \{1, \ldots, M\} \quad \text{and}\\\nonumber
   &f_j(\vect{\theta}^\star) < f_j(\vect{\theta})\quad\exists j \in \{1, \ldots, M\} .
\end{align}
\end{Definition}
The set of non-dominated solutions forms the \textit{Pareto front}, which offers trade-offs between the objectives, allowing for decision-making based on specific preferences. To measure the quality of an approximated \textit{Pareto front}, hypervolume indicator is the most commonly used metric in multi-objective optimization \cite{7360024}. Let $\tilde P$ be a Pareto front approximation in an $M$-dimensional performance space and given a reference point $r\in\mathbb{R}^M$, the hypervolume $\mathcal{V}(\tilde P)$ is defined as
\begin{align}
    \mathcal{V}(\tilde{P}) = \int_{\mathbb{R}^M} \mathds{1}_{\mathcal{V}(\tilde{P})}(z)dz.
\end{align}
The hypervolume then becomes the following set:
\begin{align}
   \mathcal{V}(\tilde{P}) =\left\{z\in Z\big|\exists 1\leq i\leq|\tilde P|:r\preceq z\preceq\tilde P(i)\right\},  
\end{align}
where $\tilde P(i)$ is the $i$-th solution in $\tilde P$, $\preceq$ is the relation operator of objective dominance, and $\mathds{1}_{\mathcal{V}(\tilde{P})}$ is a Dirac delta function that equals $1$ if $z\in\mathcal{V}(\tilde{P})$ and $0$ otherwise. To determine how much the hypervolume would increase if a set of new points $P\subset\mathbb{R}^M$ is added to the current $\tilde P$, one can use the Hypervolume Improvement (HVI) as
\begin{align}
    \text{HVI}(P,\tilde P)=\mathcal{V}(\tilde{P}\cup P)-\mathcal{V}(\tilde{P}).
\end{align}
In the context of Multi-Objective BO (MOBO), the Expected Hypervolume Improvement (EHVI) acquisition function then aims to maximize the expected increase in hypervolume as
\begin{align}
    \alpha_{\text{EHVI}}(\vect{\theta}) = \mathbb{E}[\mathcal{V}(\tilde{P}\cup P)-\mathcal{V}(\tilde{P})],
\end{align}
where each objective function in \eqref{eq:objs} is approximated by a surrogate GP model, as explained in \eqref{eq:GP}-\eqref{eq:mean_var}. 
\section{MPC-based CDPG using Multi-Objective Bayesian Optimization}\label{sec:4}
In this paper, we propose a fusion of CDPG and MOBO, leveraging the strengths of both methods to provide an effective and data-efficient learning mechanism for a differentiable parametric MPC scheme \eqref{eq:V0} with an inaccurate model. To this end, we propose to simultaneously minimize both the closed-loop performance $J\left(\vect\theta\right)$ and its gradient $ \nabla_{\vect\theta} J\left(\vect\theta\right)$, thereby achieving the necessary condition of optimality \eqref{eq:necessary_condition}. Consequently, it is necessary to evaluate these terms at every time instant to fit the acquired data to their surrogate GP models. An overview of the proposed MPC-CDPG-MOBO scheme is shown in Fig. \ref{fig0}.
\begin{figure}[htbp!]
		\centering
		\includegraphics[width=1\linewidth]{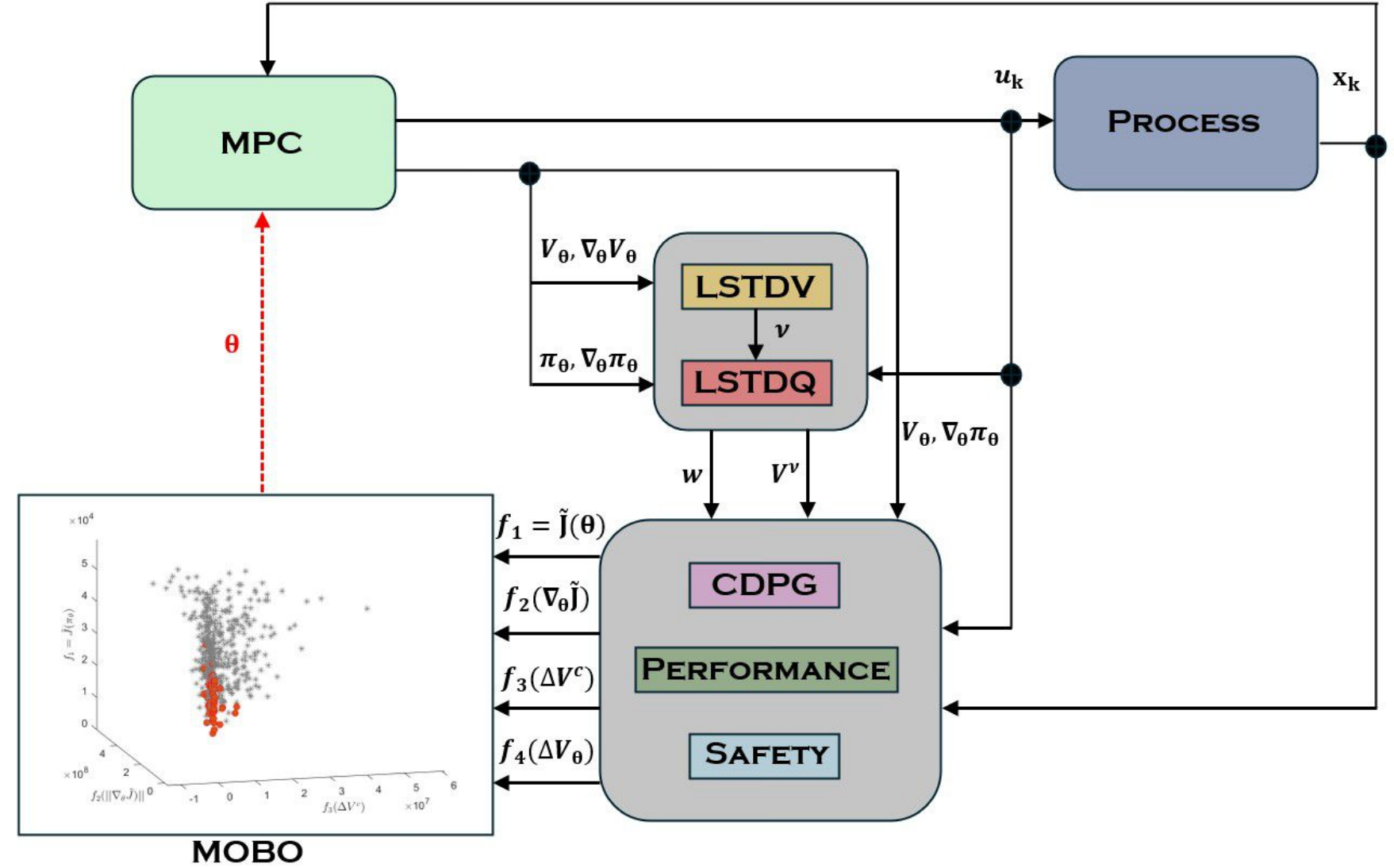}
		\caption{An overview of the proposed learning-based MPC using the fusion of reinforcement learning and multi-objective Bayesian optimization.} 
		\label{fig0}
\end{figure}
\begin{Remark}
    The proposed MPC-CDPG-MOBO framework is specifically designed to take advantage of the complete closed-loop cost information and the corresponding gradient data accumulated throughout the reinforcement learning process in every episode. This enables the method to function as a robust hybrid strategy, uniting the exploratory strength of BO (global search capability) with the fine-tuned local refinement offered by the CDPG algorithm.
\end{Remark}
To employ the closed-loop performance index $J\left(\vect\theta\right)$ as an objective function in the context of MOBO, we practically introduce a finite version of \eqref{eq:J} as 
\begin{align}\label{eq:cl_perf}
    f_1(\vect{\theta}):=\tilde J\left(\vect\theta\right)=\frac{1}{T_f}\sum_{k=0}^{T_f}\gamma^kL\left(\vect x_k,\vect \pi_{\vect\theta}(\vect x_k)\right),
\end{align}
where $\vect\pi_{\vect\theta}$ is delivered by the MPC scheme \eqref{eq:V0}. To address the necessary condition of optimality, the second objective function is defined as 
\begin{align}
    f_2(\vect{\theta})=\frac{1}{T_f}\sum_{k=0}^{T_f}\left\| \nabla _{\vect\theta}\tilde J_k(\vect\theta)\right\|,
\end{align}
where $T_f$ is the final time instant at the end of each episode, and considering \eqref{eq:PG_Q} and \eqref{eq:Q_w}, we have that
\begin{align}\label{eq:DPG}
   \nabla _{\vect\theta}\tilde J_k(\vect\theta)=\nabla_{\vect\theta}\vect\pi_{\vect\theta}\left(\vect x_k\right)\nabla_{\vect\theta}\vect\pi_{\vect\theta}^\top\left(\vect x_k\right) \vect w. 
\end{align}
Given the critic parameters $\vect w$, one then needs to compute $\nabla_{\vect\theta}\vect\pi_{\vect\theta}$ by performing a sensitivity analysis in \eqref{eq:nabla_pi}. As observed, evaluating \eqref{eq:DPG} depends on both sets of critic parameters, $\vect w$ and $\vect\nu$, delivered by the LS problem \eqref{eq:error}. To provide useful features and exploit prior knowledge for approximating the value function $V^{\vect\nu} \approx V^{\vect\pi_{\vect\theta}}$ required in \eqref{eq:Q_w}, we are inspired by \cite{10178119} and leverage the same MPC scheme \eqref{eq:V0}, used as an approximator of $\vect\pi_{\vect\theta}$, to build an approximation of $V^{\vect\pi_{\vect\theta}}$ as
\begin{align}\label{eq:nu}
    V^{\vect\nu}(\vect x_k)=V_{{\vect\theta} }(\vect x_k)+\nabla_{\vect\theta}V_{{\vect\theta} }(\vect x_k)^\top\vect\nu,
\end{align}
where the sensitivity of the parametric value function delivered by \eqref{eq:V0} is obtained as $\nabla_{\vect\theta}V_{{\vect\theta} }=\nabla_{\vect\theta} \mathcal{L}_{\vect \theta}(\vect y_k)$. To solve the LS problem \eqref{eq:error}, we use the LSTD method. To this end, an LSTDQ problem can be formulated so that \eqref{eq:error} becomes equivalent to:
\begin{align}\label{eq::LSTDQ_eq}
    &\mathbb{E}\left[\delta_Q\nabla_{\vect w} Q^{\vect w}(\vect x_k,\vect u_k)\right]=0,\\\nonumber
    &\delta_Q=\left(L\left(\vect x_k,\vect u_k\right)\right)+\gamma Q^{\vect w}\left({\vect x_{k+1},\vect\pi_{\vect\theta} (\vect x_{k+1})} \right)\\\nonumber
    &\qquad\qquad\qquad\qquad\qquad\qquad-Q^{\vect w}\left({\vect x_k,\vect u_k}\right).
\end{align}
Substituting the state-action value function $Q^{\vect w}=A^{\vect w}+V^{\vect\nu}$ into \eqref{eq::LSTDQ_eq}, we then have that
\begin{align}
    &\mathbb{E}\left[\delta\nabla_{\vect w} A^{\vect w}\left({\vect x^k,\vect u^k} \right)\right]=0,\\\nonumber
    &\delta=L\left(\vect x^k,\vect u_i^k\right)+\gamma \left(A^{\vect w}\left({\vect x_{k+1},\vect\pi_{\theta}(\vect x_{k+1})} \right)+V^{\vect\nu}\left({\vect x_{k+1}} \right)\right)\\\nonumber
    &\qquad\qquad\qquad\qquad\qquad\qquad-\left(A^{\vect w}\left({\vect x_k,\vect u_k}\right)+V^{\vect\nu}\left({\vect x_{k}} \right)\right),
\end{align}
where the advantage function $A^{\vect w}$ is obtained using a linear approximator in with the state-action features $\vect\Psi_k={{\left( {\vect u_k - {\vect\pi_{\vect\theta}\left(\vect x_k\right) }} \right)}^{\top}}\nabla_{\vect\theta}\vect\pi_{\vect\theta}^\top\left(\vect x_k\right)$. Since $A^{\vect w}\left({\vect x_{k+1},\vect\pi_{\theta}(\vect x_{k+1})} \right)=0$, the update rule of $\vect w$ is then obtained as
\begin{subequations}
\label{eq:lstdq_update}
\begin{align}
    &\vect w=\Xi_{\vect w}^{-1}b_{\vect w},\quad \Xi_{\vect w}= \sum_{k=1}^{T_f}\vect \Psi_k\vect \Psi_k^{\top},\\
    &b_{\vect w}=\sum_{k=1}^{T_f}\delta^V_k\vect\Psi_k,
\end{align}
\end{subequations}
where
\begin{align}
    \delta^V_k=L\left(\vect x_k,\vect u_k\right)+\gamma V_{\vect\theta}\left({\vect x_{k+1}}\right)-V^{\vect\nu}\left({\vect x_k}\right).
\end{align}
To obtain the update rule of $\vect\nu$ associated with $V^{\vect\nu}$ in \eqref{eq:nu}, we use an LSTDV learning method such that the LS problem \eqref{eq:error} at $\vect u_k=\vect\pi_{\theta}(\vect x_k)$ becomes equivalent to:
\begin{align}\label{eq:LSTDV}
    \mathbb{E}\left[\delta^V_k\nabla_{\vect\nu} V^{\vect\nu}\left({\vect x_k} \right)\right]=0,
\end{align}
The corresponding update rule then reads as
\begin{subequations}
\label{eq:lstdv_update}
\begin{align}
    &\vect \nu=\Xi_{\vect \nu}^{-1}b_{\vect \nu},\\
    &\Xi_{\vect \nu}=\sum_{k=1}^{T_f}\nabla_{\vect\theta}V_{{\vect\theta} }(\vect x_k)\nabla_{\vect\theta}V_{{\vect\theta} }(\vect x_k)^\top,\\
    &b_{\vect \nu}=\sum_{k=1}^{T_f}{\nabla_{\vect\theta}V_{{\vect\theta} }(\vect x_k)\left(L(\vect x_k,\vect u_k)+\gamma V_{{\vect\theta} }(\vect x_{k+1})-V_{{\vect\theta} }(\vect x_k)\right)}.
\end{align}
\end{subequations}
As observed, the updating rules of the critic parameters $\vect w,\vect\nu$ are influenced by the policy parameters $\vect\theta$ through $V_{\vect\theta}$ and $\nabla_{\vect\theta} V_{\vect\theta}$. Therefore, some policy parameters may cause a distortion in the updating rules of the critic parameters. To address this issue, we introduce the third objective function $f_3$ to be minimized using the proposed MOBO method, which forces the policy parameters to avoid an increase in the average value of $V^{\vect\nu}$ in \eqref{eq:nu} across two consecutive learning episodes, $\ell=1,\cdots,m$. 
\begin{align}
    f_3(\vect\theta)=\max\left(0,V^c_{\ell+1}-V^c_\ell\right),
\end{align}
where $V_\ell^c=\sum_{k=1}^{T_f}  V^{\vect\nu}(\vect x_k)$. Ensuring stability is a fundamental requirement for any control system. However, when employing BO to learn the parameters of an MPC scheme, obtaining formal stability guarantees is not straightforward. To address this issue, one can integrate Lyapunov stability criteria under the control policy captured from \eqref{eq:V0} directly into the learning process \cite{HIRT2024208}. To this end,  we propose building a GP model of the forth objective function $f_4$ to be embedded in the proposed MOBO-based learning mechanism.  
\begin{align}
    f_4(\vect\theta)=\sum_{k=1}^{T_f}\beta\max\left(0,V_{\vect\theta}(\vect x_{k+1})-V_{\vect\theta}(\vect x_{k})\right),\quad\beta>0.
\end{align}
\begin{Remark}
   It is worth noting that the proposed MPC-CDPG-MOBO algorithm introduces a safe Multi-Objective Reinforcement Learning (MORL) framework, wherein both the critic and actor networks are guided by the MOBO mechanism. The corresponding parameters $\vect\theta$ are updated based on multiple objective functions, $f_1(\vect\theta), \ldots, f_4(\vect\theta)$. Specifically, the method employs an MPC scheme as a shared approximator for both the critic and actor networks, thereby enabling the MOBO algorithm to update them simultaneously.
\end{Remark}

\section{Simulation Results and Discussions}\label{sec:5}
In this section, the proposed learning-based MPC scheme is applied to a Continuous Stirred Tank Reactor (CSTR), where the system dynamics are nonlinear and may not be modeled accurately. 
\subsection{CSTR Model}
In this chemical reactor, the reaction $(A\rightarrow B)$ occurs via an irreversible and exothermic process. The objective is to control the concentration of species $A$, denoted by $C_a$, and the reaction volume $V$, by manipulating the output process flow rate $q_s$ and the coolant flow rate $q_c$ \cite{PIPINO2021107195}. 
The dynamics of the CSTR are then described as 
\begin{subequations}\label{eq:CSTR_model}
		\begin{align}
         &\dot{V}(t)=q_o-q_s(t),\\
         &\dot{C}_a(t)=\frac{q_o}{V(t)}\left(C_{a_o}-C_a(t)\right)-k_0{\rm e}^{\frac{-E}{RT(t)}}C_a(t),\\
         &\dot{T}(t)=\frac{q_o}{V(t)}\left(T_o-T(t)\right)+k_1{\rm e}^{\frac{-E}{RT(t)}}C_a(t)\\\nonumber
         &\qquad\qquad\qquad\qquad+k_2\frac{q_c(t)}{V(t)}\left(1-{\rm e}^{\frac{-k_3}{q_c(t)}}\right)\left(T_{co}-T(t)\right),\\
         &k_1=\frac{-\Delta Hk_0}{\rho C_p},\quad k_2=\frac{\rho_c C_{pc}}{\rho C_p},\quad k_3=\frac{hA}{\rho_c C_{pc}} 
		\end{align}
\end{subequations}
where $V(t),C_a(t),T(t)$ are the reaction volume, the concentration of $A$, and the reactor temperature, respectively. 
\begin{figure}[htbp!]
		\centering
		\includegraphics[width=.8\linewidth]{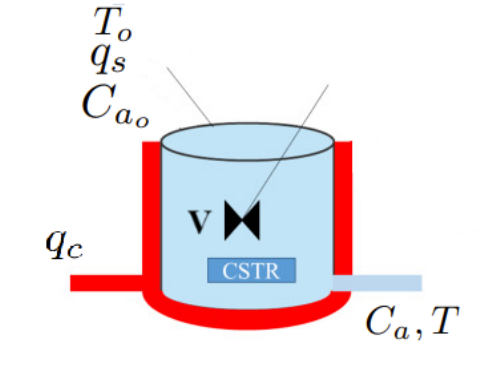}
		\caption{A basic schematic of the Continuous Stirred Tank Reactor (CSTR).} 
\end{figure}
The constant parameters listed in Table 1 include the process flow rate $q_o$, feed concentration $C_{a_o}$, reaction rate $k_0$, activation energy term $E/R$, feed temperature $T_o$, inlet coolant temperature $T_{co}$, heat of reaction $\Delta H$, heat transfer coefficient $hA$, liquid densities $\rho,\rho_c$ and specific heats $C_p, C_{pc}$.
\begin{table}[!ht]
    \centering
    \caption{CSTR Model Parameters}
    \begin{tabular}{|l|l|l|l|}
    \hline
        $q_o$ & $100[\frac{l}{min}]$ &$C_{a_o}$ & $1[\frac{mol}{l}]$ \\ \hline
        $T_o$ & $350[K]$ & $T_{co}$ & $350[K]$ \\ \hline
        $\Delta H$ & $-2\times 10^5 [\frac{cal}{mol}]$ & $\rho C_p$ & $1000[\frac{cal}{lK}]$ \\ \hline
        $k_0$ & $7.2\times10^{10}[\frac{1}{min}]$ & $E/R$ & $1\times 10^4[K]$ \\ \hline
        $\rho_cC_{pc}$ & $1000[\frac{cal}{lK}]$ & $hA$ & $7\times 10^5[\frac{cal}{minK}]$ \\ \hline
    \end{tabular}
\end{table}
\subsection{Simulation Settings}
To investigate the performance of the proposed learning mechanism in the presence of model misspecification, we use an imperfect model of \eqref{eq:CSTR_model} in the MPC scheme as
\begin{subequations}\label{eq:CSTR_Wrongmodel}
		\begin{align}
         &\dot{V}(t)=q_o-q_s(t),\\
         &\dot{C}_a(t)=1.1\frac{q_o}{V(t)}\left(C_{a_o}-C_a(t)\right)-1.2k_0{\rm e}^{\frac{-E}{RT(t)}}C_a(t),\\
         &\dot{T}(t)=\frac{q_o}{V(t)}\left(T_o-T(t)\right)+1.15k_1{\rm e}^{\frac{-E}{RT(t)}}C_a(t)\\\nonumber
         &\qquad\qquad\qquad+0.9k_2\frac{q_c(t)}{V(t)}\left(1-{\rm e}^{\frac{-1.2k_3}{q_c(t)}}\right)\left(T_{co}-T(t)\right)
		\end{align}
\end{subequations}
The constraints on the states and control inputs are $90\leq V\leq 110$, $0\leq C_a\leq 0.35$, $400\leq T\leq 480$, $55\leq q_s\leq 140$ and $55\leq q_c\leq 140$.
The parameterized MPC cost function in \eqref{eq:V0} is defined as
\begin{align}
    &\theta_c+\gamma^N\Big[(\hat{\vect x}_N-\vect x_r)^\top T_{\vect\theta}(\hat{\vect x}_N-\vect x_r)+\vect \Gamma_f^\top\vect\eta_N\Big]\\\nonumber
    &+\sum_{i=0}^N\mathcal{G}_{\vect\theta}^\top[\hat{\vect x}_i,\hat{\vect u}_i]^\top+\sum_{i=0}^{N-1}\gamma^i\Big[(\hat{\vect x}_i-\vect x_r)^\top Q_{\vect\theta}(\hat{\vect x}_i-\vect x_r)\\\nonumber
    &\qquad\qquad+(\hat{\vect u}_i-\vect u_r)^\top R_{\vect\theta}(\hat{\vect u}_i-\vect u_r)+\vect \Gamma^\top\vect\eta_i\Big],
\end{align}
where $N$ is set to $10$. We choose $\vect\Gamma_f=\vect\Gamma=[10^5,10^5,10^5]^\top$. To discretize the continuous CSTR model, a fourth-order Runge-Kutta (RK4) integrator is used with a sampling time of $0.05 \text{min}$. In the learning setting, the parameterized MPC is updated over $600$ episodes, each consisting of $3\text{min}$ ($60$ time steps) of interaction with the real system. The initial conditions are randomly selected for each episode. The reference points $\vect x_r,\vect u_r$ are defined as $V^d=105$ $l$, $C_a^d=0.12$ $mol/l$, $T^d=433$ $K$, $q_s^d=100$ $l/min$ and $q_c^d=110$ $l/min$. We define the RL stage cost function as
\begin{align}
    &L\left(\vect x_k,\vect u_k\right)=\\\nonumber
    &\qquad(\vect x_k-\vect x_r)^\top\text{diag}([10,5000,10])(\vect x_k-\vect x_r)\\\nonumber
    &\qquad\qquad+(\vect u_k-\vect u_r)^\top\text{diag}([10,10])(\vect u_k-\vect u_r)\\\nonumber
    &\qquad\qquad\qquad\qquad\qquad+\vect\Gamma^\top\max(0,\vect h(\vect x_k,\vect u_k)).
\end{align}
\subsection{Discussions}
Fig. \ref{fig1} compares the evolution of the closed-loop performance across three methods: MPC-based Reinforcement Learning (MPC-RL), single-objective BO, and the proposed MOBO-based MPC-RL (MPC-RL-MOBO) approach. The results show that MPC-RL-MOBO achieves rapid convergence to the optimal performance $J^\star$, outperforming both baseline methods in terms of convergence speed and accuracy. Fig. \ref{fig2} further highlights the benefit of MOBO by demonstrating improved satisfaction of the necessary optimality conditions in the Compatible Deterministic Policy Gradient (CDPG) method, as indicated by the faster convergence of the advantage function $A^{\vect w}$ to zero. The accompanying 3D plot shows the evolution of the Pareto front, with red stars denoting the optimal trade-off solutions. Figures \ref{fig3} and \ref{fig4} provide state and input trajectories, respectively, showing the transition from an imperfect MPC model (green) to an adjusted MPC through learning (gray to blue). Both figures confirm the successful adaptation of the MPC parameters over time, resulting in improved tracking performance within defined state and input constraints.
\begin{figure}[htbp!]
		\centering
		\includegraphics[width=1\linewidth]{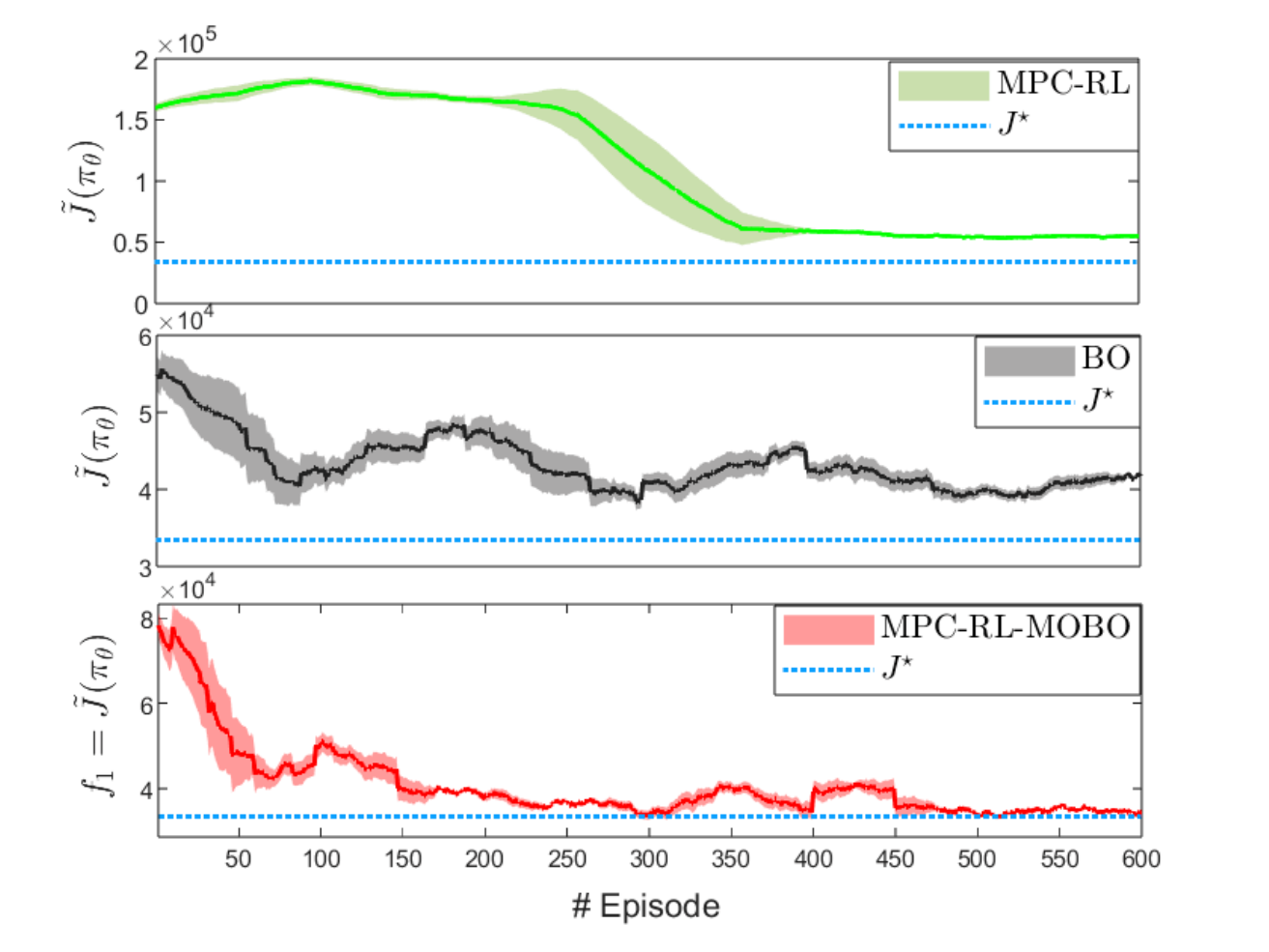}
		\caption{The green line in the first subplot illustrates the evolution of the closed-loop performance $\tilde J$ using the MPC-RL method. The gray line in the second subplot depicts the performance evolution when employing a single-objective BO method. The red line in the third subplot represents the performance under the proposed MPC-RL-MOBO. The blue dashed line indicates the closed-loop performance achieved by an MPC with a perfect model. As observed, the proposed method converges rapidly to the optimal performance $J^\star$, whereas the BO method approaches $J^\star$ with slightly lower accuracy but still relatively fast convergence. In contrast, the MPC-based RL method takes approximately $350$ episodes to approach the optimal performance with lower accuracy} 
		\label{fig1}
\end{figure}

\begin{figure}[htbp!]
		\centering
		\includegraphics[width=1\linewidth]{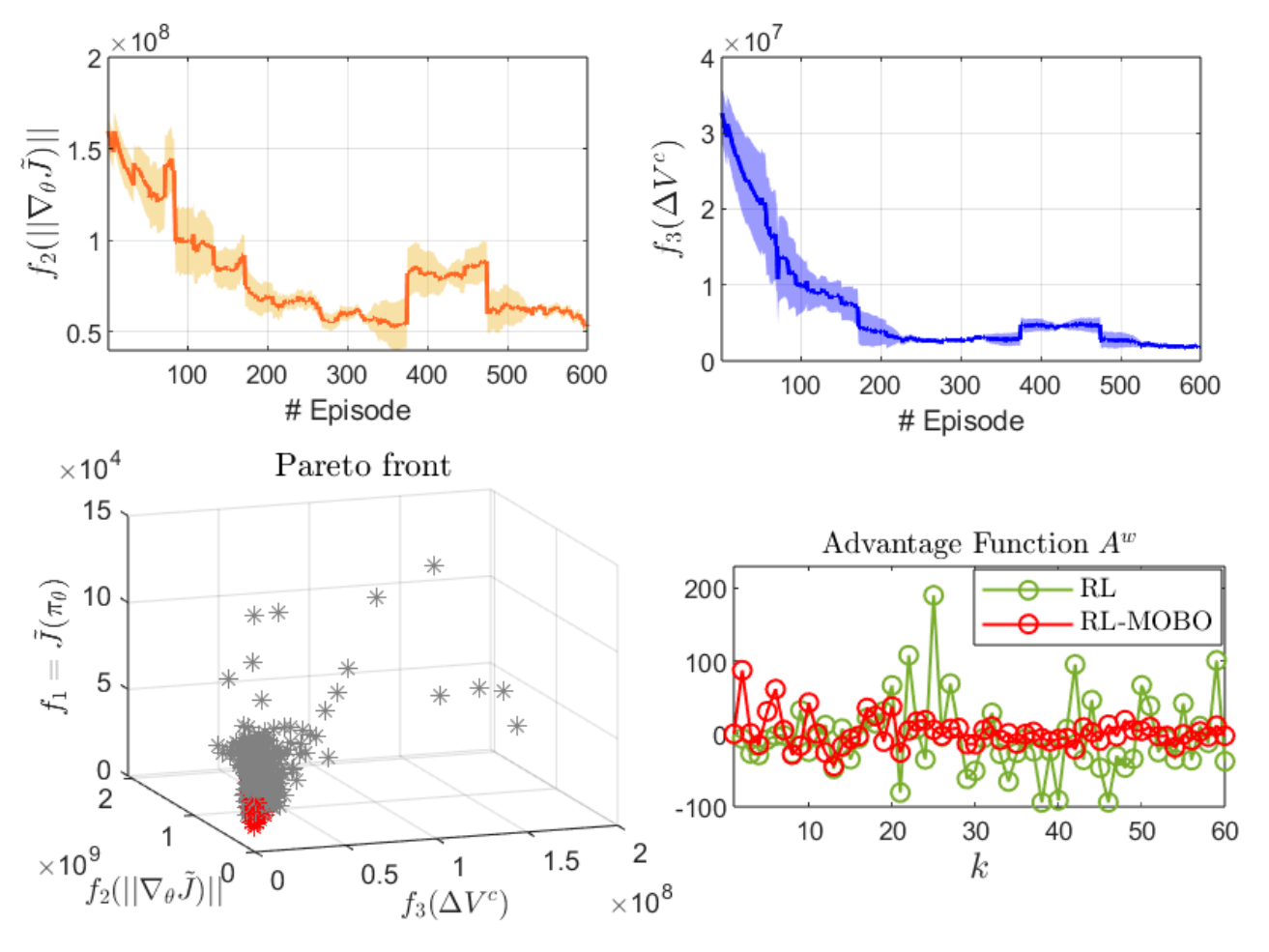}
		\caption{This figure demonstrates the effectiveness of the multi-objective Bayesian optimization (MOBO) in enhancing the performance of the MPC-based Compatible Deterministic Policy Gradient (CDPG) method. As shown, incorporating the two objective functions, $f_2,f_3$, into the MOBO framework allows for a more accurate satisfaction of the necessary optimality conditions. This is evidenced by the convergence of the advantage function $A^{\vect w}$ toward zero in the proposed MPC-RL-MOBO approach (shown in red), compared to the MPC-RL method (shown in green). The 3D plot illustrates the evolution of the \textit{Pareto front}, with the red stars indicating the optimal solutions.} 
		\label{fig2}
\end{figure}

\begin{figure}[htbp!]
		\centering
		\includegraphics[width=1\linewidth]{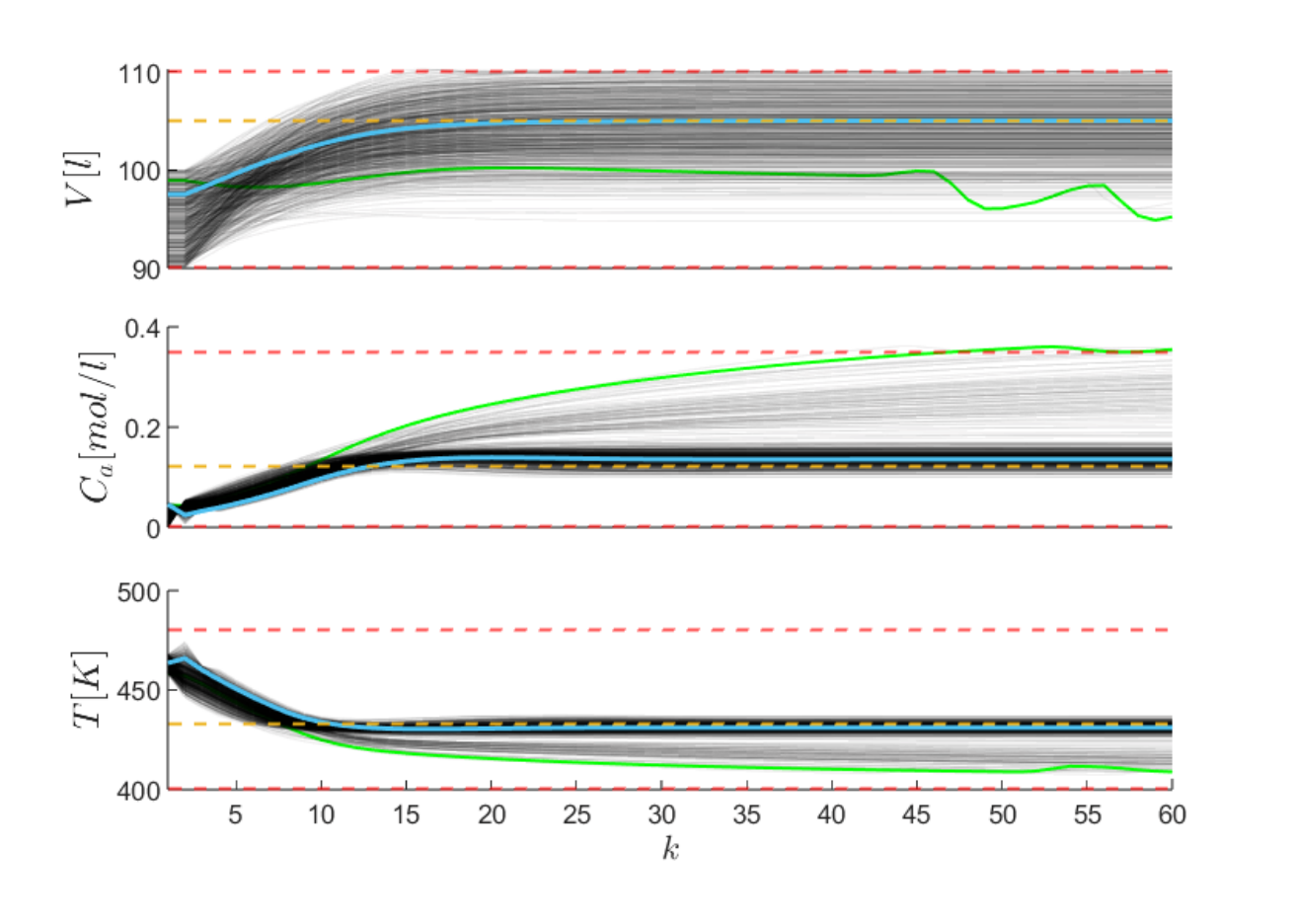}
		\caption{The green lines show the evolution of the states where an imperfect MPC model is used. We then start learning the parameterized MPC where the gray lines show the corresponding states during the learning process. Finally, the blue lines show the evolution of the states obtained using an adjusted MPC after $600$ learning episodes. The orange dashed lines are the references, and the red dashed lines illustrate the state limits.} 
		\label{fig3}
\end{figure}

\begin{figure}[htbp!]
		\centering
		\includegraphics[width=1\linewidth]{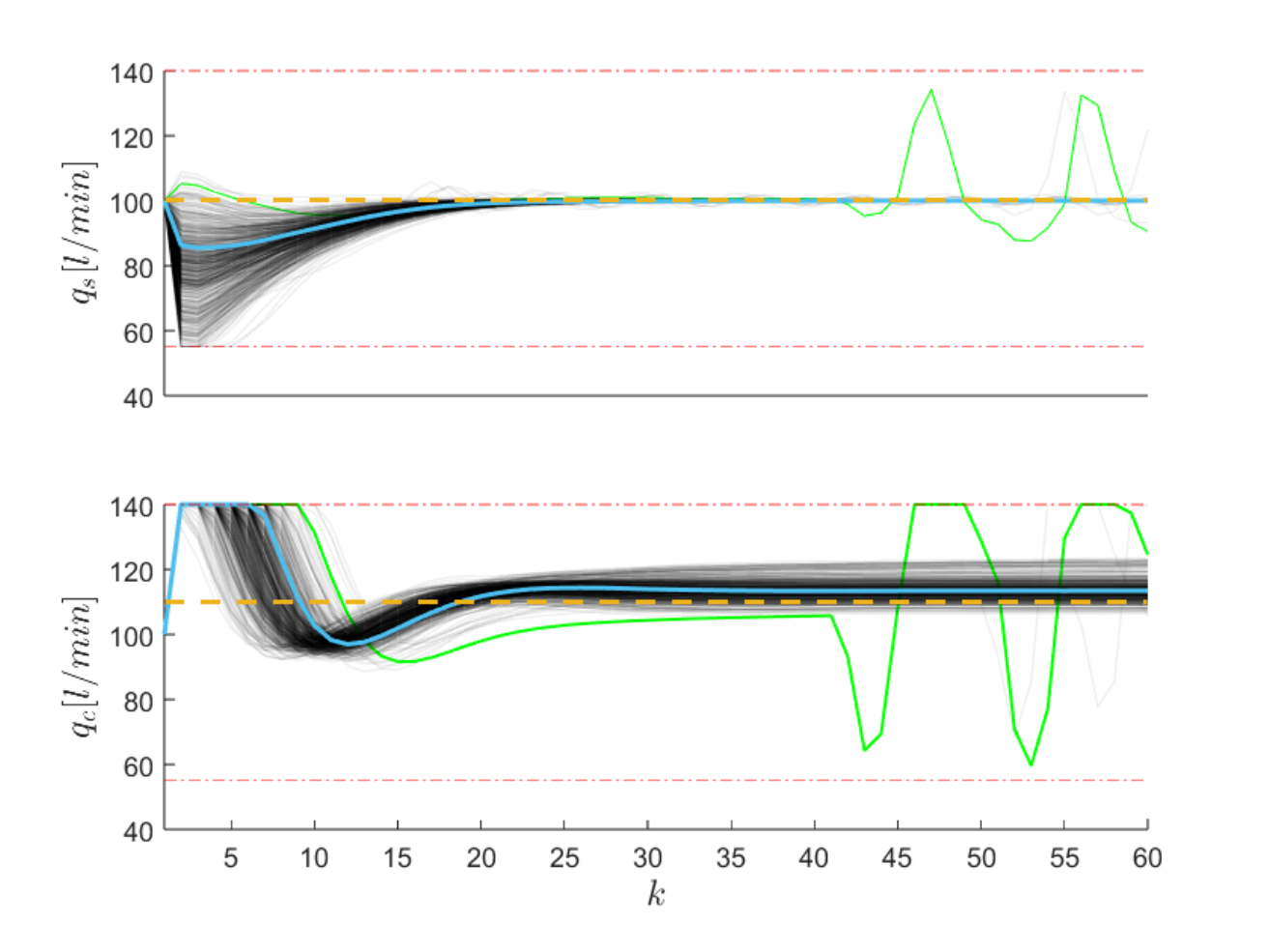}
		\caption{The green lines show the evolution of the control inputs where an imperfect MPC model is used. We then start learning the parameterized MPC where the gray lines show the corresponding control signals during the learning process. Finally, the blue lines show the evolution of the control inputs obtained using an adjusted MPC after $600$ learning episodes. The orange dashed lines are the references, and the red dashed lines show the limits.} 
		\label{fig4}
\end{figure}

\section{Conclusion}\label{sec:6}
This paper presented a novel framework that integrates Model Predictive Control (MPC)-based reinforcement learning with Multi-Objective Bayesian Optimization (MOBO) to address key limitations of existing MPC-based RL methods. Although MPC-based RL offers interpretability and lower complexity compared to deep neural network-based approaches, it suffers from slow convergence, suboptimality due to limited parameterization, and safety concerns during learning. To overcome these challenges, we introduced a sample-efficient and stability-aware learning scheme that combines a deterministic policy gradient with the MOBO framework employing the Expected Hypervolume Improvement (EHVI) acquisition function. This fusion enables adaptive tuning of the MPC parameters to improve closed-loop performance even when the underlying model is imperfect. Overall, the proposed method advances the practical applicability of MPC-based RL by enhancing convergence efficiency, safety, and robustness, paving the way for more reliable deployment in real-world control systems.
\bibliographystyle{IEEEtran}
\bibliography{IEEEabrv,ref}
\end{document}